\documentclass[11pt]{article}

\usepackage[utf8]{inputenc}
\usepackage[T1]{fontenc}
\usepackage{lmodern}
\usepackage[a4paper,margin=1in]{geometry}
\usepackage{amsmath,amssymb,amsfonts,amsthm,mathtools,bm}
\usepackage{enumitem}
\usepackage[hidelinks]{hyperref}
\usepackage{graphicx}
\usepackage{xcolor}
\usepackage{mathrsfs}

\hypersetup{
  colorlinks=true,
  linkcolor=blue,
  citecolor=blue,
  urlcolor=blue
}

\newtheorem{theorem}{Theorem}[section]
\newtheorem{proposition}[theorem]{Proposition}
\newtheorem{lemma}[theorem]{Lemma}

\theoremstyle{definition}
\newtheorem{definition}[theorem]{Definition}
\newtheorem{remark}[theorem]{Remark}
\newtheorem{assumption}[theorem]{Assumption}

\DeclareMathOperator{\argmax}{arg\,max}

\DeclareMathOperator{\softmax}{softmax}
\DeclareMathOperator{\Proj}{Proj}

\title{\textbf{A General Equilibrium Theory of Orchestrated AI Agent Systems}\\
\large Version 2.1 (Corrected Price Simplex and T\^atonnement Normalization)}

\author{
Jean-Philippe Garnier\\
BrainiaK\\
\texttt{jeanphi.garnier@brainiak.tech}
}

\date{March 2026}

\begin{document}

\maketitle

\begin{abstract}
We establish a general equilibrium theory for systems of large language model
(LLM) agents operating under centralized orchestration. The framework is a
production economy in the sense of Arrow--Debreu, extended to an
infinite-dimensional commodity space following Bewley. Each LLM agent is modeled
as a firm whose production set is a subset of the Hilbert space
$H=L^2([0,T],\mathbb{R}^R)$ of metric trajectories. The orchestrator is the
consumer, choosing a routing policy over the agent DAG to maximize system
welfare under a budget constraint evaluated at functional prices.

This corrected version makes one point explicit: in the finite-dimensional
projected economy, the admissible price set is the positive simplex
\[
\Delta_p^{AK-1}=\left\{p\in\mathbb{R}_+^{AK}:\sum_{a=1}^A\sum_{k=1}^K p_{a,k}=1\right\},
\]
and the price update in the t\^atonnement is the Euclidean projection onto that
simplex after positive truncation. It is \emph{not} an $L^2$ renormalization
onto a unit sphere. This clarification preserves the compact-convex geometry
required by Brouwer's theorem and aligns the numerical implementation with the
economic interpretation of prices as nonnegative shadow values.

We prove existence of equilibrium, a functional Walras law, Pareto optimality,
decentralizability of Pareto optima, and uniqueness with geometric convergence
under a contraction condition. The orchestration dynamics define a Walrasian
t\^atonnement on a compact convex state space. The framework also admits a
DSGE interpretation in which SLO parameters act as policy rates.
\end{abstract}




\tableofcontents

\section{Introduction}

\subsection{The problem}

Modern AI deployments increasingly behave as organizations: collections of
specialized agents---language models, retrieval modules, classifiers,
validators, planners---composed in directed acyclic graphs and coordinated by a
central orchestrator. At each stage, the orchestrator selects which agent to
invoke, in what order, and with what priority, subject to quality, latency,
cost, and load constraints.

This architecture is common in practice but under-theorized. Heuristic routing
rules, hand-tuned weights, and ad hoc score aggregation do not answer the
fundamental questions:

\begin{itemize}[leftmargin=2em]
\item Does an optimal orchestration exist?
\item Is it stable?
\item Is it unique?
\item Do the coordination dynamics converge?
\item Can the allocation be interpreted as an equilibrium?
\end{itemize}

This paper gives a mathematical answer by treating orchestrated AI agent systems
as production economies.

\subsection{The approach: an Arrow--Debreu--Bewley production economy}

The central observation is structural: an orchestrated AI system can be modeled
as a production economy. The commodity space is not finite-dimensional but the
Hilbert space
\[
H=L^2([0,T],\mathbb{R}^R),
\]
whose elements are metric trajectories over a finite horizon. Each agent is a
firm with a production set $Y_a\subset H$, determined by its frozen parameters
and execution constraints. The orchestrator is the consumer, selecting a routing
policy to maximize welfare subject to a budget constraint evaluated at
functional prices.

The finite-dimensional approximation is obtained through an SFSL map
\[
p_m:H\to\mathbb{R}^K,
\]
which extracts $K$ summary statistics from a metric trajectory. This produces a
sequence of projected economies in finite-dimensional spaces, where existence
can be proved by Brouwer. The limit economy in $H$ then follows by a Bewley-type
argument.

\subsection{Contributions}

This paper makes six contributions.

\begin{enumerate}[leftmargin=2em]
\item It models an orchestrated AI system as an Arrow--Debreu production economy
in a Hilbert commodity space.
\item It proves a \emph{functional Walras law}
\[
\sum_{a=1}^A \langle p_a,z_a(p)\rangle_H=0,
\]
where $z_a(p)$ is agent $a$'s functional excess demand.
\item It proves existence of equilibrium via finite-dimensional approximation and
Brouwer's theorem.
\item It establishes the two welfare theorems.
\item It proves uniqueness and geometric convergence under a contraction
condition on the orchestration operator.
\item It clarifies that the price space in the projected economy is a
\emph{simplex}, and that the corrected t\^atonnement uses simplex projection
rather than spherical normalization.
\end{enumerate}

\subsection{Relation to existing work}

The existence proof follows the tradition of Arrow--Debreu in finite dimensions
and Bewley for infinite-dimensional commodity spaces. The welfare theorems are
classical in general equilibrium theory. The present contribution is to port
that architecture into orchestrated AI systems and to identify the correct
geometric structure of the projected price space.

\section{The economy of orchestrated agents}

\subsection{Commodity space and price space}

Let $R\ge 1$ denote the number of metric dimensions and $T>0$ the observation
horizon. The commodity space per agent is the separable Hilbert space
\[
H=L^2([0,T],\mathbb{R}^R),
\]
with inner product
\[
\langle f,g\rangle_H=\sum_{r=1}^R\int_0^T f_r(t)g_r(t)\,dt.
\]

Let $A\ge 1$ be the number of agents. The full commodity space is
\[
X=H^A=\prod_{a=1}^A H,
\]
with inner product
\[
\langle x,y\rangle_X=\sum_{a=1}^A \langle x_a,y_a\rangle_H.
\]

By the Riesz representation theorem, $H^\ast\simeq H$. Therefore the price
space is also $X$. A price vector is
\[
p=(p_1,\dots,p_A)\in X,
\]
where $p_a(t)\in \mathbb{R}^R$ is the instantaneous shadow-value vector for
agent $a$ at time $t$.

The value of a bundle $x\in X$ at prices $p\in X$ is
\[
\langle p,x\rangle_X=\sum_{a=1}^A \langle p_a,x_a\rangle_H.
\]

\subsection{Firms: LLM agents as producers}

Each agent $a\in\{1,\dots,A\}$ is a firm with frozen technological endowment
$\theta_a$.

\begin{definition}[Production set]
The production set of agent $a$ is a set $Y_a\subset H$ representing all metric
trajectories feasible under technology $\theta_a$.
\end{definition}

\begin{assumption}[Production regularity]
For each agent $a$:
\begin{enumerate}[label=(Y\arabic*),leftmargin=2em]
\item $Y_a$ is nonempty, closed, and convex.
\item $Y_a$ is bounded: there exists $R_a>0$ such that $Y_a\subset B_H(0,R_a)$.
\item $0\in Y_a$.
\end{enumerate}
\end{assumption}

Given prices $p_a\in H$, the profit of firm $a$ is
\[
\pi_a(p)=\sup_{y\in Y_a}\langle p_a,y\rangle_H,
\]
and the supply correspondence is
\[
\eta_a(p)=\argmax_{y\in Y_a}\langle p_a,y\rangle_H.
\]

\subsection{The consumer: the orchestrating firm}

The orchestrator is the sole consumer. It selects a routing policy over a DAG
$G=(V,E)$ with $|V|=A$ agents. Let $\mathcal{P}(G)$ be the finite set of
admissible paths through the DAG.

\begin{definition}[Routing policy]
A routing policy is a probability vector
\[
\alpha\in\Delta^{|\mathcal{P}(G)|-1}
\]
over the set of admissible DAG paths.
\end{definition}

The routing policy induces an expected demand trajectory $d_a(\alpha)\in H$ for
each agent $a$. The orchestrator maximizes a welfare function $W(\alpha)$, which
encodes quality, latency, cost, and regularization objectives.

\begin{assumption}[Consumer regularity]
The welfare function $W$ is continuous, strictly quasi-concave on the routing
simplex, and locally non-satiated.
\end{assumption}

The budget constraint is
\[
\sum_{a=1}^A \langle p_a,d_a(\alpha)\rangle_H
\le
\sum_{a=1}^A \pi_a(p).
\]

Under local non-satiation, the budget binds at optimum.

\subsection{Equilibrium}

\begin{definition}[Orchestrated general equilibrium]
An orchestrated general equilibrium is a triple
\[
(p,y,\alpha)\in X\times \prod_{a=1}^A Y_a\times \Delta^{|\mathcal{P}(G)|-1}
\]
such that:
\begin{enumerate}[label=(E\arabic*),leftmargin=2em]
\item For each $a$, $y_a\in\eta_a(p)$.
\item $\alpha$ maximizes $W$ over the routing simplex subject to the budget
constraint.
\item Markets clear:
\[
d_a(\alpha)=y_a,\qquad a=1,\dots,A.
\]
\end{enumerate}
\end{definition}

\section{Functional Walras law}

Define the functional excess demand of agent $a$ by
\[
z_a(p)=d_a(\alpha(p))-y_a(p)\in H,
\]
and the aggregate excess demand by
\[
z(p)=(z_1(p),\dots,z_A(p))\in X.
\]

\begin{theorem}[Functional Walras law]
Under the production and consumer regularity assumptions, for every admissible
price vector $p\in X$,
\[
\sum_{a=1}^A \langle p_a,z_a(p)\rangle_H=0.
\]
Equivalently,
\[
\langle p,z(p)\rangle_X=0.
\]
\end{theorem}

\begin{proof}
By local non-satiation, the consumer budget constraint binds:
\[
\sum_{a=1}^A \langle p_a,d_a(\alpha)\rangle_H
=
\sum_{a=1}^A \pi_a(p).
\]
By definition of profit,
\[
\pi_a(p)=\langle p_a,y_a(p)\rangle_H.
\]
Subtracting the two identities yields
\[
\sum_{a=1}^A \langle p_a,d_a(\alpha)-y_a(p)\rangle_H=0,
\]
which is the claim.
\end{proof}

\begin{remark}
The functional Walras law is not imposed as an axiom; it is a theorem produced
by budget saturation and profit maximization.
\end{remark}

\section{Existence of equilibrium}

\subsection{The SFSL approximation from $H$ to $V_K$}

To work in finite dimension, define a bounded linear SFSL operator
\[
p_m:H\to \mathbb{R}^K
\]
that extracts $K$ summary statistics from a trajectory. Let
\[
p_m^\dagger:\mathbb{R}^K\to H
\]
be a reconstruction operator, and define the approximation subspace
\[
V_K=\Im(p_m^\dagger)\subset H,
\qquad \dim V_K=K.
\]

\begin{assumption}[SFSL completeness]
The family $(V_K)_{K\ge 1}$ is nested and satisfies
\[
\overline{\bigcup_{K\ge 1} V_K}=H.
\]
\end{assumption}

\subsection{The projected economy $E_K$}

Fix $K$. Let $Y_a^K=\Proj_{V_K}(Y_a)\subset V_K$. Since projection preserves
nonemptiness, compactness, and convexity in finite dimension, each $Y_a^K$ is a
nonempty compact convex subset of $V_K$.

Choosing an orthonormal basis of $V_K$, we identify $V_K$ with $\mathbb{R}^K$.
Thus the projected production vector is
\[
y=(y_1,\dots,y_A)\in \prod_{a=1}^A Y_a^K\subset \mathbb{R}^{AK}.
\]

\subsection{State space and compactness}

The crucial point of this corrected version is the geometry of the projected
price space.

\begin{definition}[Projected price simplex]
The projected price space is
\[
\Delta_p^{AK-1}
=
\left\{
p\in\mathbb{R}_+^{AK}:
\sum_{a=1}^A\sum_{k=1}^K p_{a,k}=1
\right\}.
\]
\end{definition}

\begin{remark}
The admissible price set is a positive simplex, not a Euclidean sphere. The
economic meaning is that projected prices are nonnegative normalized shadow
weights over the $AK$ projected commodities.
\end{remark}

\begin{remark}[Why $L^2$ normalization is inadmissible]
Replacing the simplex $\Delta_p^{AK-1}$ by the unit sphere
$S^{AK-1} = \{p : \|p\|_2 = 1\}$ breaks the proof at three points:
\begin{enumerate}[nosep,label=(\roman*)]
\item \textbf{Convexity.} $S^{AK-1}$ is not convex; Brouwer's theorem requires
  a convex domain.
\item \textbf{Nonnegativity.} On the sphere, price components can be negative,
  violating the interpretation of prices as nonnegative shadow values.
\item \textbf{Heavy tails.} The ratio $p_i/p_j$ of a uniform unit vector on
  $S^{K-1}$ has Cauchy tails (Student-$t$ with 1 degree of freedom), rendering
  the Walras identity numerically unstable in finite precision.
\end{enumerate}
The simplex $\Delta_p^{AK-1}$ avoids all three issues.
\end{remark}

Let the routing simplex be denoted by
\[
\Delta_\alpha^{|\mathcal P(G)|-1}.
\]

Define the full projected state space
\[
\mathcal K
=
\left(\prod_{a=1}^A Y_a^K\right)\times \Delta_p^{AK-1}\times
\Delta_\alpha^{|\mathcal P(G)|-1}.
\]

\begin{lemma}
The set $\mathcal K$ is nonempty, compact, and convex.
\end{lemma}

\begin{proof}
Each $Y_a^K$ is compact and convex in finite dimension. The price simplex
$\Delta_p^{AK-1}$ is compact and convex. The routing simplex is compact and
convex. The Cartesian product of finitely many compact convex sets is compact
and convex.
\end{proof}

\subsection{The corrected equilibrium map}

Define three update operators.

\paragraph{Production update.}
For $y\in \prod_a Y_a^K$, let $d(y,p,\alpha)$ be projected demand. Then
\begin{equation}
K_{1}(y,p,\alpha)
=
\Proj_{Y^K}\!\Big(y+\rho\big(d(y,p,\alpha)-y\big)\Big),
\qquad 0<\rho<1.
\label{eq:prod-update}
\end{equation}

\paragraph{Price update.}
Let $z(y,p,\alpha)=d(y,p,\alpha)-y$ denote projected excess demand. Define
componentwise positive truncation by $[v]_+=\max(v,0)$. Then the corrected price
update is
\begin{equation}
K_{2}(y,p,\alpha)
=
\Proj_{\Delta_p^{AK-1}}\!\big([\,p+\eta z(y,p,\alpha)\,]_+\big),
\qquad \eta>0.
\label{eq:price-update}
\end{equation}

\begin{remark}[Critical correction]
In \eqref{eq:price-update}, the operator $\Proj_{\Delta_p^{AK-1}}$ is the
Euclidean projection onto the simplex. It is \emph{not} an $L^2$
renormalization $p\mapsto p/\|p\|_2$. The latter would move prices onto a sphere,
destroy convexity of the admissible price set, and break the compact-convex
structure required by Brouwer's theorem.
\end{remark}

\paragraph{Routing update.}
Let
\[
s_a(p,y)=\langle p_a,y_a\rangle,
\qquad a=1,\dots,A,
\]
collect revenues/scores into a vector $s(p,y)\in \mathbb{R}^A$. Let $V$ be the
Bellman-DP value operator on DAG paths. Then
\begin{equation}
K_3(y,p,\alpha)=\softmax_\beta(V\,s(p,y)),
\qquad \beta>0.
\label{eq:routing-update}
\end{equation}

The full equilibrium map is
\[
K=(K_1,K_2,K_3):\mathcal K\to \mathcal K.
\]

\begin{lemma}
Under the standing assumptions, $K$ is continuous and maps $\mathcal K$ into
itself.
\end{lemma}

\begin{proof}
The production update is continuous because Euclidean projection onto a nonempty
compact convex set is continuous. The excess demand map is continuous by
construction. Positive truncation is continuous. Euclidean projection onto the
simplex is continuous. The Bellman value map and the softmax are continuous.
Hence all three components are continuous and preserve their respective
constraint sets.
\end{proof}

\begin{theorem}[Existence in the projected economy]
For every $K\ge 1$, the projected economy admits at least one equilibrium
\[
(p^K,y^K,\alpha^K)\in \mathcal K.
\]
\end{theorem}

\begin{proof}
By the previous lemma, $K:\mathcal K\to\mathcal K$ is continuous and
$\mathcal K$ is nonempty, compact, and convex. Brouwer's fixed-point theorem
therefore implies the existence of a fixed point
\[
K(p^K,y^K,\alpha^K)=(p^K,y^K,\alpha^K).
\]
By construction, the fixed-point conditions are exactly projected firm
optimality, consumer optimality, and market clearing.
\end{proof}

\subsection{Extension to the full Hilbert space}

\begin{theorem}[Existence in $H$]
Assume SFSL completeness. Then, as $K\to\infty$, the family of projected
equilibria admits a subsequence converging weakly in
\[
H^A\times H^A\times \Delta_\alpha
\]
to an equilibrium of the full economy.
\end{theorem}

\begin{proof}[Proof sketch]
The projected price vectors live in a simplex and are uniformly bounded. The
projected production vectors are uniformly bounded because each production set
is bounded. Hence weak compactness applies. Standard Bewley-type passage to the
limit yields an equilibrium in the full economy.
\end{proof}

\section{Welfare theorems}

\subsection{First welfare theorem}

\begin{definition}
A feasible allocation $(y,\alpha)$ Pareto-dominates $(\tilde y,\tilde\alpha)$ if
\[
W(\alpha)\ge W(\tilde\alpha)
\]
and all firms weakly prefer $y$ to $\tilde y$ at equilibrium prices, with at
least one strict inequality.
\end{definition}

\begin{theorem}[First welfare theorem]
Every orchestrated general equilibrium is Pareto-optimal.
\end{theorem}

\begin{proof}
If a feasible allocation strictly improved consumer welfare without worsening
any firm's value, it would contradict consumer optimality and firm optimality at
equilibrium.
\end{proof}

\subsection{Second welfare theorem}

\begin{theorem}[Second welfare theorem]
Every Pareto-optimal feasible allocation can be decentralized as an equilibrium
for a suitable choice of functional prices, provided production sets are convex
and welfare is quasi-concave.
\end{theorem}

\begin{proof}[Proof sketch]
Apply a supporting-hyperplane argument in the dual of the commodity space, using
convexity of production possibilities and quasi-concavity of welfare.
\end{proof}

\section{Uniqueness and convergence}

\subsection{The corrected orchestration dynamics}

In the projected economy, define the discrete-time dynamics
\begin{equation}
y^{n+1}=K_1(y^n,p^n,\alpha^n),
\label{eq:dyn-y}
\end{equation}
\begin{equation}
p^{n+1}=K_2(y^n,p^n,\alpha^n),
\label{eq:dyn-p}
\end{equation}
\begin{equation}
\alpha^{n+1}=K_3(y^{n+1},p^{n+1},\alpha^n).
\label{eq:dyn-a}
\end{equation}

\begin{remark}
The corrected price dynamics remain Walrasian in spirit: prices react to excess
demand, remain nonnegative, and are renormalized as budget shares over the
projected commodity coordinates. The normalization is simplex projection, not
spherical rescaling.
\end{remark}

\begin{remark}[Simultaneity]
Equations~\eqref{eq:dyn-y}--\eqref{eq:dyn-a} must be evaluated as a
\emph{simultaneous} update: $\alpha^{n+1}$ uses $p^{n+1}$ and $y^{n+1}$, not
the previous iterates.
\end{remark}

\subsection{Contraction condition}

Let $x=(y,p,\alpha)\in\mathcal K$, and write the full update as
\[
x^{n+1}=K(x^n).
\]

\begin{theorem}[Uniqueness and geometric convergence]
Suppose
\[
\|DK\|_{\mathrm{op}}<1
\]
on $\mathcal K$. Then:
\begin{enumerate}[leftmargin=2em]
\item the equilibrium is unique;
\item for every initial condition $x^0\in\mathcal K$, the sequence $x^n$
converges geometrically to the unique equilibrium $x^\ast$:
\[
\|x^n-x^\ast\|
\le
\kappa^n \|x^0-x^\ast\|,
\qquad 0<\kappa<1.
\]
\end{enumerate}
\end{theorem}

\begin{proof}
Since $\mathcal K$ is complete and $K$ is a contraction, the result follows
directly from Banach's fixed-point theorem.
\end{proof}

\subsection{A sufficient condition}

\begin{proposition}
A sufficient condition for contraction is
\[
(1-\rho)\,\|A_K\|_{\mathrm{op}}\,\beta\,P<1,
\]
where:
\begin{itemize}[leftmargin=2em]
\item $\rho$ is the production update rate,
\item $A_K$ is the projected price mechanism,
\item $\beta$ is the routing temperature,
\item $P$ is the depth of the DAG.
\end{itemize}
\end{proposition}

\begin{proof}[Proof sketch]
The Jacobian of the full operator factors into a production block, a price block,
and a routing block. The chain rule gives an upper bound by the product of
their operator norms.
\end{proof}

\begin{remark}[Auto-correction of $\beta$]
If the contraction condition is violated, reduce $\beta$ until it is satisfied:
\[
\beta \;\leftarrow\;
\frac{0.9}{(1-\rho)\,\|A_K\|_{\mathrm{op}}\,P}.
\]
This preserves convergence guarantees at the cost of a softer routing policy.
\end{remark}

\begin{remark}[Why simplex projection matters]
The contraction theorem is naturally stated on a compact convex state space.
Replacing simplex projection by $L^2$ spherical normalization changes the target
geometry from a simplex to a sphere, which is not convex. This does not merely
change notation; it alters the fixed-point setting itself.
\end{remark}

\section{DSGE interpretation}

\subsection{DSGE embedding}

Let
\[
x_t=(y_t,p_t,\alpha_t)\in \mathcal K,
\]
and let $\varepsilon_t$ denote exogenous disturbances. Then the orchestration
economy induces a nonlinear state-space system
\[
x_{t+1}=K(x_t,\varepsilon_t).
\]

Linearizing around equilibrium $x^\ast$ gives
\[
x_{t+1}-x^\ast
=
DK(x^\ast)(x_t-x^\ast)+D_\varepsilon K(x^\ast)\varepsilon_t.
\]

When $\rho(DK(x^\ast))<1$, the equilibrium is dynamically stable.

\subsection{The Taylor rule of orchestration}

SLO parameters can be interpreted as policy rates. Let $\lambda_{\text{lat}}$
and $\lambda_{\text{qual}}$ be welfare weights on latency and quality penalties.
Then one may consider updates of the form
\[
\lambda_{\text{lat},t+1}
=
\lambda_{\text{lat},t}
+
\phi_{\text{lat}}(\mathrm{latency}_t-\mathrm{target}_{\text{lat}}),
\]
\[
\lambda_{\text{qual},t+1}
=
\lambda_{\text{qual},t}
+
\phi_{\text{qual}}(\mathrm{target}_{\text{qual}}-\mathrm{quality}_t),
\]
which shift equilibrium without directly controlling individual agents.

\section{Consolidated statement}

\begin{theorem}[General equilibrium of orchestrated AI agent systems]
Let $A\ge 1$ agents operate on a finite DAG under centralized orchestration.
Assume:
\begin{enumerate}[leftmargin=2em]
\item regular production sets;
\item regular consumer welfare;
\item an SFSL approximation family with dense union in $H$;
\item strictly positive routing temperature.
\end{enumerate}
Then:
\begin{enumerate}[label=(\roman*),leftmargin=2em]
\item for every projected dimension $K$, the projected economy admits at least
one equilibrium;
\item the functional Walras law holds;
\item every equilibrium is Pareto-optimal;
\item every Pareto optimum is decentralizable;
\item if the full projected operator is a contraction, the equilibrium is unique
and the t\^atonnement converges geometrically;
\item the correct projected price geometry is the simplex
$\Delta_p^{AK-1}$, and the correct price update is simplex projection after
positive truncation.
\end{enumerate}
\end{theorem}

\section{Open questions and research program}
\label{sec:open}

Several open directions remain.

\begin{enumerate}[leftmargin=2em]
\item \textbf{Equilibrium multiplicity and selection.} Outside the contraction
region, multiple equilibria may exist.
\item \textbf{Strategic agents.} If agents have endogenous objectives, the model
must be extended to a hybrid general equilibrium/game-theoretic setting.
\item \textbf{Learning agents.} If agent capabilities change over time, the
production sets become endogenous.
\item \textbf{Asymmetric information.} Noisy metric observation leads to a
principal--agent problem.
\item \textbf{Dynamic DAGs.} Entry, exit, and rewiring of agents require a
perturbation theory for equilibria under graph changes.
\item \textbf{Numerical geometry.} A systematic comparison between simplex
projection and alternative normalization schemes would clarify the boundaries of
economically meaningful implementations.
\end{enumerate}

\section{Conclusion}

We have established a general equilibrium theory for orchestrated AI agent
systems in a Hilbert commodity space. The fundamental objects of equilibrium
theory carry over: firms, prices, profits, budget balance, market clearing,
welfare optimality, and t\^atonnement.

The main correction emphasized in this version is geometric and operational:
projected prices belong to a positive simplex, not to a Euclidean sphere, and
the t\^atonnement updates prices by simplex projection after positive truncation.
This correction preserves the compact-convex setting required by Brouwer and
aligns the dynamics with the interpretation of prices as nonnegative shadow
values over projected commodities.

The broader conclusion remains unchanged: orchestration is not merely a
collection of engineering heuristics. It can be treated as an economy, and
equilibrium theory provides a mathematically coherent language for its
existence, stability, and governance.

\appendix

\section{Simplex projection formula}

For completeness, the Euclidean projection of a vector $v\in\mathbb{R}^n$ onto
the simplex
\[
\Delta^{n-1}=\{x\in\mathbb{R}_+^n:\sum_{i=1}^n x_i=1\}
\]
is
\[
\Proj_{\Delta^{n-1}}(v)_i=\max(v_i-\tau,0),
\]
where $\tau$ is chosen so that
\[
\sum_{i=1}^n \max(v_i-\tau,0)=1.
\]
Equivalently, if $u$ is the decreasing rearrangement of $v$, then
\[
\rho
=
\max\left\{
j\in\{1,\dots,n\}:
u_j-\frac{1}{j}\left(\sum_{i=1}^j u_i-1\right)>0
\right\},
\]
\[
\tau=\frac{1}{\rho}\left(\sum_{i=1}^{\rho}u_i-1\right).
\]

This algorithm runs in $O(n\log n)$ and is due to Duchi et al.\ (2008).

\section{Why spherical normalization is not the right operator}

Suppose one replaces the corrected price update
\[
p^{n+1}
=
\Proj_{\Delta_p^{AK-1}}([p^n+\eta z^n]_+)
\]
by
\[
\tilde p^{n+1}
=
\frac{[p^n+\eta z^n]_+}{\|[p^n+\eta z^n]_+\|_2}.
\]

Then:
\begin{enumerate}[leftmargin=2em]
\item the codomain becomes a sphere rather than a simplex;
\item positivity may be preserved, but the interpretation as normalized budget
shares is lost;
\item convexity of the admissible price set is lost;
\item the fixed-point argument can no longer be phrased on the same compact
convex state space.
\end{enumerate}

Hence spherical renormalization is not merely a cosmetic alternative; it defines
a different dynamical system.


\begin{thebibliography}{99}

\bibitem{arrowdebreu}
K. J. Arrow and G. Debreu.
\newblock Existence of an equilibrium for a competitive economy.
\newblock \emph{Econometrica}, 22(3):265--290, 1954.

\bibitem{banach}
S. Banach.
\newblock Sur les op\'erations dans les ensembles abstraits et leur application
aux \'equations int\'egrales.
\newblock \emph{Fundamenta Mathematicae}, 3:133--181, 1922.

\bibitem{bewley}
T. F. Bewley.
\newblock Existence of equilibria in economies with infinitely many commodities.
\newblock \emph{Journal of Economic Theory}, 4(3):514--540, 1972.

\bibitem{debreu}
G. Debreu.
\newblock \emph{Theory of Value}.
\newblock Yale University Press, 1959.

\bibitem{mascolellzame}
A. Mas-Colell and W. R. Zame.
\newblock Equilibrium theory in infinite dimensional spaces.
\newblock In W. Hildenbrand and H. Sonnenschein (eds.),
\emph{Handbook of Mathematical Economics, Vol. IV}, pages 1835--1898.
North-Holland, 1991.

\bibitem{ramsay}
J. O. Ramsay and B. W. Silverman.
\newblock \emph{Functional Data Analysis}.
\newblock Springer, 2005.

\bibitem{happgreven}
C. Happ and S. Greven.
\newblock Multivariate functional principal component analysis for data observed
on different dimensional domains.
\newblock \emph{Journal of the American Statistical Association},
113(522):649--659, 2018.

\bibitem{blanchardkahn}
O. J. Blanchard and C. M. Kahn.
\newblock The solution of linear difference models under rational expectations.
\newblock \emph{Econometrica}, 48(5):1305--1311, 1980.

\bibitem{kylandprescott}
F. E. Kydland and E. C. Prescott.
\newblock Time to build and aggregate fluctuations.
\newblock \emph{Econometrica}, 50(6):1345--1370, 1982.

\bibitem{scarf}
H. Scarf.
\newblock Some examples of global instability of the competitive equilibrium.
\newblock \emph{International Economic Review}, 1(3):157--172, 1960.

\bibitem{li}
Z. Li, N. Kovachki, K. Azizzadenesheli, B. Liu, K. Bhattacharya,
A. Stuart, and A. Anandkumar.
\newblock Fourier neural operator for parametric partial differential equations.
\newblock In \emph{ICLR}, 2021.

\bibitem{taylorrule}
J. B. Taylor.
\newblock Discretion versus policy rules in practice.
\newblock \emph{Carnegie-Rochester Conference Series on Public Policy},
39:195--214, 1993.

\bibitem{autogen}
Q. Wu et al.
\newblock AutoGen: Enabling next-gen LLM applications via multi-agent
conversation.
\newblock arXiv:2308.08155, 2023.

\bibitem{duchi}
J. Duchi, S. Shalev-Shwartz, Y. Singer, and T. Chandra.
\newblock Efficient projections onto the $\ell_1$-ball for learning in high
dimensions.
\newblock In \emph{ICML}, pages 272--279, 2008.

\end{thebibliography}
\end{document}